\documentclass[a4paper,10pt]{article}
\usepackage[footnotesize]{caption}
\usepackage[caption=false,font=footnotesize]{subfig}
\usepackage{mathptmx,graphicx,amsmath,amsfonts,amsthm,amssymb,url,cite}
\usepackage{a4wide}

\newcommand{\expn}{E}
\newcommand{\pr}{P}
\newcommand{\eps}{\epsilon}

\newtheorem{definition}{Definition}
\newtheorem{theorem}{Theorem}
\newtheorem{lemma}{Lemma}
\newtheorem{corollary}{Corollary}

\begin{document}

\title{Optimal symmetric Tardos traitor tracing schemes}
\author{Thijs Laarhoven\footnote{T. Laarhoven and B. de Weger are with the Department of Mathematics and Computer Science, Eindhoven University of Technology, 5612 AZ Eindhoven, The Netherlands.\protect\\
E-mail: \{t.m.m.laarhoven,b.m.m.d.weger\}@tue.nl
This work was done when the first author was with Irdeto, Eindhoven, The Netherlands. The content is mostly based on the first author's Master's thesis \cite{msc11}.} \and Benne de Weger\footnotemark[1]}
\date{\today}
\maketitle

\begin{abstract}
For the Tardos traitor tracing scheme, we show that by combining the symbol-symmetric accusation function of \v{S}kori\'{c} et al.\ with the improved analysis of Blayer and Tassa we get further improvements. Our construction gives codes that are up to $4$ times shorter than Blayer and Tassa's, and up to $2$ times shorter than the codes from \v{S}kori\'{c} et al. Asymptotically, we achieve the theoretical optimal codelength for Tardos' distribution function and the symmetric score function. For large coalitions, our codelengths are asymptotically about $4.93\%$ of Tardos' original codelengths, which also improves upon results from Nuida et al.
\end{abstract}

\section{Introduction}
\label{sec:Introduction}

Watermarking digital content allows distributors of copyrighted digital data to embed so-called fingerprints into their data in such a way that each copy of the data can be uniquely identified. These watermarks are made in a robust way, so that users cannot change or remove them from the content. If a copy of the data is then illegally distributed to unauthorized users and intercepted by the distributor, he can extract the fingerprint from the copy and find the person whose fingerprinted data was distributed. Actions can then be taken against this user, to prevent further illegal distribution.

To be able to trace the watermarked data back to the user, we need that the embedded fingerprints for each user are different. However, by comparing their differently watermarked copies of the content, multiple malicious users can form a coalition and detect differences in their content. Assuming that besides the watermarks all copies are the same, this allows coalitions to detect part of the watermark. By editing this data, they can then create a forged copy, which contains the same digital content as their original copies, but has a forged fingerprint that cannot be traced back to them directly. Under the marking assumption, which says that colluders can only detect and edit fingerprint positions if their fingerprints do not all match on that position, there are ways to construct fingerprinting schemes such that any forged copy can be traced back to at least one of the colluders. This involves finding a construction for fingerprints for each of the users, and finding a way to trace back forged copies to guilty users.

\subsection{Model}
\label{sub:Introduction-Model}

Let $U = \{1, \ldots, n\}$ denote the set of the $n$ users that received watermarked content. Here a user corresponds to one watermarked copy of the content, so a person who possesses several differently watermarked copies of the data is assumed to control multiple users. For each user $j$ the distributor generates a fingerprint (also called a codeword), which is usually denoted by $\vec{x}_j$. This codeword is a vector of length $\ell$ (the codelength) of symbols from an alphabet $Q$ of size $q$. The case $q = 2$ corresponds to the binary alphabet, which is usually taken as $Q = \{0,1\}$. All fingerprints together form the fingerprinting code $\mathcal{C} = \{\vec{x}_1, \ldots, \vec{x}_n\}$. A common way of representing this code is by putting all codewords as rows in a matrix $X$ according to $X_{ji} = (\vec{x}_j)_i$.

After assigning codewords to users and distributing the watermarked copies, a subset $C \subseteq U$ of $c$ users (called colluders or pirates) may form a coalition to create a forged copy. Using some pirate strategy $\rho$, a function $Q^{\ell \times c} \to Q^{\ell}$, they construct a forged copy, which has some unknown distorted fingerprint $\rho(X) = \vec{y}$ called the forgery. For the pirate strategy $\rho$, we assume that the marking assumption holds, i.e. if for all $j \in C$ the pirates have $(\vec{x}_j)_i = \omega$ for some position $i$ and symbol $\omega \in Q$, then the coalition is forced to output $y_i = \omega$. On other positions, we assume that colluders are free to choose any of the symbols from the alphabet.

Finally, after the coalition has created a forged copy, we assume the distributor intercepts it and extracts the forgery $\vec{y}$ from the data. He then runs some tracing algorithm $\sigma$ on the forgery, to get a subset $\sigma(\vec{y}) \subseteq U$ of users that are accused. The accusation is said to be successful if no innocent users are accused (i.e. $\sigma(\vec{y}) \subseteq C$) and at least one guilty user is accused (i.e. $\sigma(\vec{y}) \cap C \neq \emptyset$).

In the setting of probabilistic schemes, the code $X$ and the tracing algorithm $\sigma$ may depend on some random variables. The events of not accusing any innocent users (soundness) and accusing at least one guilty user (completeness) then also depend on these random variables. Then, instead of demanding that a fingerprinting scheme is always sound and complete, we may demand that the probability of failure is bounded by some small value $\eps$, where the probability is taken over these random variables. This leads to the following definitions of $\eps_1$-soundness and $\eps_2$-completeness.

\begin{definition}[Soundness and completeness] \label{def:Secureness}
Let $C \subseteq U$ be a coalition of size at most $c$, and let $\rho$ be some pirate strategy employed by this coalition. Then a traitor tracing scheme $(X, \sigma)$ is called $\eps_1$-sound if
  \begin{align*}
  \pr[\sigma(\rho(X)) \not\subseteq C] \leq \eps_1.
  \end{align*}
Similarly, a fingerprinting scheme is called $\eps_2$-complete if
  \begin{align*}
  \pr[\sigma(\rho(X)) \cap C = \emptyset] \leq \eps_2.
  \end{align*}
\end{definition}

As we will see later, $\eps_1/n$ and $\eps_2$ are closely related in the Tardos fingerprinting scheme. Therefore it is convenient to introduce the notation $\eta = \log(\eps_2)/\log(\eps_1/n)$ such that $\eps_2 = (\eps_1/n)^{\eta}$, which describes how big $\eps_2$ is, compared to $\eps_1/n$. Also, we sometimes simply say a scheme is secure, to denote that it is sound and complete for certain (implicit) parameters $\eps_1$ and $\eps_2$.

\subsection{Related work}
\label{sub:Introduction-RelatedWork}

\sloppypar{In \cite{tardos03}, Tardos investigated probabilistic binary fingerprinting schemes where small margins of error are allowed. He proved that a codelength of $\ell = \Omega(c^2 \ln(n/\eps_1))$ is necessary to achieve soundness and completeness, while in the same paper he also gave a construction with a codelength of $\ell = 100 c^2 \ln(n/\eps_1)$. This construction is often referred to as the Tardos scheme. In \cite{amiri09,huang09a} the lower bound on the codelength was further tightened, to show that one needs $\ell \geq 2 \ln(2) c^2 \ln(n/\eps_1)$ for sufficiently large $c$ and $q = 2$, to achieve soundness and completeness.}

Since the scheme of Tardos had a constant $100$ in front of the $c^2 \ln(n/\eps_1)$ in the codelength, many papers focused on constructing a scheme with the same order codelength, but with a smaller constant. For example, using a discrete distribution function in the Tardos scheme, Nuida et al.\ showed in \cite{nuida09} that one can achieve codelengths of $\ell < 5 c^2 \ln(n/\eps_1)$ in some cases with small $c$, while for large $c$ they achieved an asymptotic codelength of $\ell \approx 5.35 c^2 \ln(n/\eps_1)$. Using a different approach, Amiri and Tardos showed in \cite{amiri09} that with a computation-heavy construction, one can approach the theoretical lower bound of $\ell = 2 \ln(2) c^2 \ln(n/\eps_1)$ for large $c$.

In this paper we will focus on the binary Tardos scheme with the arcsine distribution function from \cite{tardos03}, which was introduced in \cite{tardos03} and further analyzed and improved in e.g. \cite{blayer08,nuida09,skoric08,skoric06}. We will focus on two improvements in particular. In \cite{blayer08}, Blayer and Tassa made the proofs of \cite{tardos03} tighter by introducing several auxiliary variables which were to be optimized later, instead of fixing them in advance. In that paper the construction of the Tardos scheme essentially remained the same, but it was shown that a codelength of $\ell = 85 c^2 \ln(n/\eps_1)$ is also sufficient to prove soundness and completeness. In \cite{skoric08}, \v{S}kori\'{c} et al.\ did change the scheme, by making the score function of the Tardos scheme symbol-symmetric. This also lead to shorter codelengths, giving asymptotic codelengths of $\ell = (\pi^2 + o(1)) c^2 \ln(n/\eps_1) \approx 9.87 c^2 \ln(n/\eps_1)$ for large $c$, while maintaining soundness and completeness. Furthermore assuming that the scores of innocent users and the joint coalition score are normally distributed, \v{S}kori\'{c} et al.\ showed in \cite[Section 6]{skoric08} that an asymptotic codelength of $\ell = (\frac{\pi^2}{2} + o(1)) c^2 \ln(n/\eps_1)$ is then both sufficient and necessary. Since by the Central Limit Theorem these scores will in fact converge to normal distributions for asymptotically large $c$, this also provides a lower bound on the codelength, when using the arcsine distribution function and the symmetric score function.

\subsection{Contributions and outline}
\label{sub:Introduction-Contributions}

Combining the symbol-symmetric score function from \v{S}kori\'{c} et al.\ with Blayer and Tassa's sharp analysis, we will prove $\eps_1$-soundness and $\eps_2$-completeness for all $c \geq 2$ and $\eta \leq 1$ with a codelength of $\ell = 23.79 c^2 \ln(n/\eps_1)$. This improves upon the codelength from Blayer and Tassa by a factor more than $3.5$, and it improves upon the original Tardos scheme by a factor of more than $4$. Furthermore, for bigger $c$ and smaller $\eta$ the constant in front of the $c^2 \ln(n/\eps_1)$ in $\ell$ further decreases, easily leading to a factor $10$ improvement over the original Tardos scheme and a factor slightly less than $4$ improvement over the Blayer and Tassa analysis.

Similar to work of \v{S}kori\'{c} et al., we also look at the asymptotics of our scheme, and show that for large $c$, we can prove soundness and completeness for a codelength of $\ell = (\frac{\pi^2}{2} + O(c^{-1/3})) c^2 \ln(n/\eps_1) \approx 4.93 c^2 \ln(n/\eps_1)$. This improves upon the asymptotic results from \v{S}kori\'{c} et al.\ by a factor $2$, and we achieve the asymptotic optimal codelength which \v{S}kori\'{c} et al.\ proved to be sufficient and necessary under the added assumption that the distributions of scores are normal distributions. We therefore close the gap of a factor $2$ between the best known provably secure codelength and the asymptotic optimal codelength, for Tardos' original arcsine distribution function and the symmetric score function. These results also improve upon the asymptotic codelengths from Nuida et al., who used different discrete distribution functions, by more than $7\%$.

The paper is organized as follows. In Section \ref{sec:Construction} we first give the construction of the (symmetric) Tardos scheme, and compare our results with earlier results from literature. In Sections \ref{sec:Soundness} and \ref{sec:Completeness} we then prove that the soundness and completeness properties hold under our assumptions on the parameters. In Section \ref{sec:Optimization} we then give results similar to those in \cite[Section 2.4.5]{blayer08} on how to find the optimal set of parameters that satisfies the conditions for our proof method to work, and minimizes the codelength. There we also give such minimal codelengths, for several values of $c$ and $\eta$. Finally in Section \ref{sec:Asymptotics} we prove the results stated above for asymptotically large $c$, and show that the optimal rate of convergence is of order $O(c^{-1/3})$.

This paper is mainly based on results from the first author's Master's thesis \cite{msc11}.

\section{Construction and results}
\label{sec:Construction}

First we present the construction of the Tardos traitor tracing scheme, as in \cite{blayer08}, where we use auxiliary variables $d_{\ell}, d_z, d_{\delta}$ for the codelength $\ell$, accusation offset $Z$ and cutoff parameter $\delta$ respectively. The only difference between our construction and that of Blayer and Tassa is in the score function we use. While Blayer and Tassa used the asymmetric score function from Tardos' original scheme, we use the symbol-symmetric score function from \v{S}kori\'{c} et al.

\subsection{The Tardos traitor tracing scheme}
\label{sub:Construction-Tardos}

Let $n \geq c \geq 2$ be positive integers, and let $\eps_1, \eps_2 \in (0,1)$ be the desired upper bounds for the soundness and completeness error probabilities respectively. Let us write $k = \ln(n/\eps_1)$ so that $e^{-k} = \eps_1/n$. Let $d_{\ell}, d_z, d_{\delta}$ be positive constants, with $d_{\delta} > 1$. Then the symmetric Tardos fingerprinting scheme works as follows.

\begin{enumerate}
  \item \textbf{Initialization}
  \begin{enumerate}
    \item Take the codelength as $\ell = d_{\ell} c^2 k$. \footnote{Note that $\ell$ may not be integral, while the codelength of a code of course has to be integral. See Appendix \ref{sec:IntegralCodelengths} for a short note on how to solve this minor problem in our construction.}
		\item Take the accusation offset parameter as $Z = d_z c k$.
		\item Take the cutoff parameter as $\delta = 1/(d_{\delta} c)$, and compute $\delta' = \arcsin(\sqrt{\delta})$ such that $0 < \delta' < \pi/4$.
    \item For each fingerprint position $1 \leq i \leq \ell$, select $p_i \in [\delta, 1 - \delta]$ independently from the distribution defined by the following CDF $F(p)$ and PDF $f(p)$:
      \begin{align}
      F(p) = \frac{2 \arcsin(\sqrt{p}) - 2\delta'}{\pi - 4\delta'}, \quad f(p) = \frac{1}{(\pi - 4\delta')\sqrt{p (1 - p)}}. \label{dist}
      \end{align}
    The function $f(p)$ is biased towards $\delta$ and $1 - \delta$ and symmetric around $1/2$.
  \end{enumerate}
  \item \textbf{Codeword generation}
  \begin{enumerate}
    \item For each position $1 \leq i \leq \ell$ and for each user $1 \leq j \leq n$, select the $i$th entry of the codeword of user $j$ according to $\pr[X_{ji} = 1] = p_i$ and $\pr[X_{ji} = 0] = 1 - p_i$.
  \end{enumerate}
  \item \textbf{Accusation}
  \begin{enumerate}
    \item For each position $1 \leq i \leq \ell$ and for each user $1 \leq j \leq n$, calculate the score $S_{ji}$ according to:
      \begin{align}
      S_{ji} = \begin{cases}
        +\sqrt{(1 - p_i)/p_i} & \text{if $X_{ji} = 1, y_i = 1$}, \\
        -\sqrt{p_i/(1 - p_i)} & \text{if $X_{ji} = 0, y_i = 1$}, \\
        -\sqrt{(1 - p_i)/p_i} & \text{if $X_{ji} = 1, y_i = 0$}, \\
        +\sqrt{p_i/(1 - p_i)} & \text{if $X_{ji} = 0, y_i = 0$}.
        \end{cases} \label{scoresym}
      \end{align}
    \item For each user $1 \leq j \leq n$, calculate the total accusation sum $S_j = \sum_{i = 1}^{\ell} S_{ji}$. User $j$ is accused if and only if $S_j > Z$.
  \end{enumerate}
\end{enumerate}

Under certain conditions on the parameters $d_{\ell}, d_z, d_{\delta}$, which are specified in Subsections \ref{sub:Construction-AsymTardos} and \ref{sub:Construction-SymTardos}, one can prove soundness and completeness, using (a modified version of) Tardos' proof construction. Note that, since this proof method uses several non-tight bounds, it is very well possible that there exist sets of parameters that do not satisfy these conditions, but still guarantee soundness and completeness. So if the conditions are not satisfied, we can only conclude that the proof method does not work in that case.

\subsection{Results for the asymmetric Tardos scheme}
\label{sub:Construction-AsymTardos}

In the original Tardos scheme, and in several papers discussing the Tardos scheme, the score function is asymmetric in $y_i$, as only the positions with $y_i = 1$ are taken into account for the accusations. The construction of this asymmetric Tardos scheme is the same as in Subsection \ref{sub:Construction-Tardos}, but with the scores from \eqref{scoresym} replaced by:
\begin{align}
  S_{ji} = \begin{cases}
    +\sqrt{(1 - p_i)/p_i} & \text{if $X_{ji} = 1, y_i = 1$}, \\
    -\sqrt{p_i/(1 - p_i)} & \text{if $X_{ji} = 0, y_i = 1$}, \\
    0 & \text{otherwise}.
    \end{cases} \label{scoreasym}
\end{align}
Blayer and Tassa performed an extensive analysis of this scheme in \cite{blayer08}, and showed that under the following assumptions, one can prove soundness and completeness for given $c$ and $\eta$. In these Theorems, the function $h^{-1}: (0, \infty) \to (\frac{1}{2}, \infty)$ is defined by $h^{-1}(x) = (e^x - 1 - x)/x^2$, while the function $h: (\frac{1}{2}, \infty) \to (0, \infty)$ denotes its inverse function as in \cite{blayer08}, so that $e^x \leq 1 + x + \lambda x^2$ for all $x \leq h(\lambda)$.

\begin{theorem} \cite[Theorem 1.1]{blayer08}
Let the Tardos scheme be constructed as in Subsection \ref{sub:Construction-Tardos}, but with the asymmetric score function from \eqref{scoreasym}. Let $d_{\alpha}, r$ be positive constants, with $r > \frac{1}{2}$, such that $d_{\ell}, d_z, d_{\delta}, d_{\alpha}$ and $r$ satisfy the following two requirements:
\begin{align*}
d_{\alpha} & \geq \frac{\sqrt{d_{\delta}}}{h(r) \sqrt{c}}, \tag{S1} \label{req1} \\
\frac{d_z}{d_{\alpha}} - \frac{r d_{\ell}}{d_{\alpha}^2} & \geq 1. \tag{S2} \label{req2}
\end{align*}
Then the scheme is $\eps_1$-sound.
\end{theorem}

\begin{theorem} \cite[Theorem 1.2]{blayer08}
Let the Tardos scheme be constructed as in Subsection \ref{sub:Construction-Tardos}, but with the asymmetric score function from \eqref{scoreasym}. Let $s, g$ be positive constants such that $d_{\ell}, d_z, d_{\delta}, s$ and $g$ satisfy the following two requirements:
\begin{align*}
\frac{1 - \frac{2}{d_{\delta}}}{\pi} - \frac{h^{-1}(s) s}{\sqrt{d_{\delta} c}} & \geq g, \tag{C1} \label{btreq3} \\
g d_{\ell} - d_z & \geq \eta \sqrt{\frac{d_{\delta}}{s^2 c}}. \tag{C2} \label{req4}
\end{align*}
Then the scheme is $\eps_2$-complete.
\end{theorem}

Tardos' original choice of parameters was the following, which allowed him to prove his scheme is $\eps_1$-sound and $\eps_2$-complete for all $c \geq 2$ and $\eta \leq \sqrt{c}/4$ \cite[Theorems 1 and 2]{tardos03}:
\begin{align*}
d_{\ell} = 100, \quad d_z = 20, \quad d_{\delta} = 300, \quad d_{\alpha} = 10, \quad r = 1, \quad s = 1, \quad g = \frac{1}{4}.
\end{align*}
Blayer and Tassa proved that to achieve $\eps_1$-soundness and $\eps_2$-completeness for all $c \geq 2$ and $\eta \leq 1$, the following choice of parameters is also provably secure \cite[Section 2.4]{blayer08}:
\begin{align*}
d_{\ell} = 85, \quad d_z = 15, \quad d_{\delta} = 40, \quad d_{\alpha} = 8, \quad r = 0.611, \quad s = 0.757, \quad g = 0.2461.
\end{align*}
In \cite[Corollary 1]{skoric06}, \v{S}kori\'{c} et al.\ showed that the following choice of parameters suffices to prove soundness and completeness for asymptotically large $c$:
\begin{align*}
d_{\ell} \to 4\pi^2, \quad d_z \to 4\pi, \quad d_{\delta} \to \infty, \quad d_{\alpha} \to 2\pi, \quad r = 1, \quad s = h(1), \quad g \to \frac{1}{\pi}.
\end{align*}
According to the Central Limit Theorem, the scores of innocent users and the total score of the coalition converge to certain normal distributions. Under the assumption that the scores behave exactly like these normal distributions, \v{S}kori\'{c} et al.\ showed in \cite[Corollary 3]{skoric06} that the following choice of parameters is then sufficient and necessary to prove soundness and completeness:
\begin{align*}
d_{\ell} \to 2\pi^2, \quad d_z \to 2\pi, \quad d_{\delta} \to \infty.
\end{align*}
Applying the analysis from Section \ref{sec:Asymptotics} to the asymmetric Tardos scheme, we can prove that the following choice of parameters is provably sufficient for large $c$:\footnote{These results can be obtained by applying the analysis from Section \ref{sec:Asymptotics} to Blayer and Tassa's original analysis for the asymmetric Tardos scheme. The main difference is that then one needs $g = \frac{1}{\pi} + o(1)$ instead of $g = \frac{2}{\pi} + o(1)$, which causes an extra factor $4$ for $d_{\ell}$ and extra factors $2$ for $d_z$ and $d_{\alpha}$.}
\begin{align*}
d_{\ell} \to 2\pi^2, \quad d_z \to 2\pi, \quad d_{\delta} \to \infty, \quad d_{\alpha} \to \pi, \quad r \to \frac{1}{2}, \quad s \to \infty, \quad g \to \frac{1}{\pi}.
\end{align*}
So with Blayer and Tassa's proof construction, we obtain a $2$ times shorter asymptotic codelength compared to the shortest provable codelength of \v{S}kori\'{c} et al.\ for the asymmetric Tardos scheme, and we achieve the asymptotic optimal codelength for the asymmetric Tardos scheme which \v{S}kori\'{c} et al.\ only achieved when they added the assumption that scores behave like normal distributions.

\subsection{Results for the symmetric Tardos scheme}
\label{sub:Construction-SymTardos}

We will prove in Sections \ref{sec:Soundness} and \ref{sec:Completeness} that with the following assumptions on the parameters, we can also prove soundness and completeness for the symmetric Tardos scheme.

\begin{theorem} \label{thm:Soundness}
Let the Tardos scheme be constructed as in Subsection \ref{sub:Construction-Tardos}, and let $d_{\alpha}, r$ be positive constants, with $r > \frac{1}{2}$, such that $d_{\ell}, d_z, d_{\delta}, d_{\alpha}$ and $r$ satisfy the requirements from \eqref{req1} and \eqref{req2}. Then the scheme is $\eps_1$-sound.
\end{theorem}

\begin{theorem} \label{thm:Completeness}
Let the Tardos scheme be constructed as in Subsection \ref{sub:Construction-Tardos}, and let $s, g$ be positive constants, such that $d_{\ell}, d_z, d_{\delta}, s$ and $g$ satisfy \eqref{req4} and the following requirement:
\begin{align*}
\frac{2 - \frac{4}{d_{\delta}}}{\pi} - \frac{h^{-1}(s) s}{\sqrt{d_{\delta} c}} & \geq g. \tag{C1'} \label{req3}
\end{align*}
Then the scheme is $\eps_2$-complete.
\end{theorem}

Using the above results, in Section \ref{sec:Optimization} we will prove $\eps_1$-soundness and $\eps_2$-completeness for all $c \geq 2$ and $\eta \leq 1$ for the following set of parameters:
\begin{align*}
d_{\ell} = 23.79, \quad d_z = 8.06, \quad d_{\delta} = 28.31, \quad d_{\alpha} = 4.58, \quad r = 0.67, \quad s = 1.07, \quad g = 0.49.
\end{align*}
This improves upon the constants from Blayer and Tassa by a factor more than $3.5$, and it improves upon the original Tardos scheme by a factor more than $4$. Furthermore, for bigger $c$ and smaller $\eta$ the values of $d_{\ell}$ further decrease, easily leading to a factor $10$ improvement over the original Tardos scheme.

\v{S}kori\'{c} et al. showed that for asymptotically large $c$, the following set of parameters is sufficient for proving soundness and completeness in the symmetric Tardos scheme \cite[Corollary 1]{skoric08}:
\begin{align*}
d_{\ell} \to \pi^2, \quad d_z \to 2\pi, \quad d_{\delta} \to \infty, \quad d_{\alpha} \to \pi, \quad r = 1, \quad s = h(1), \quad g \to \frac{2}{\pi}.
\end{align*}
With the added assumption that the scores of innocent users and the joint score of guilty users are normally distributed, \v{S}kori\'{c} et al.\ also showed that the following set of parameters is sufficient for soundness and completeness, for asymptotically large $c$ \cite[Corollary 2]{skoric08}:
\begin{align*}
d_{\ell} \to \frac{\pi^2}{2}, \quad d_z \to \pi, \quad d_{\delta} \to \infty.
\end{align*}
Since by the Central Limit Theorem these scores will also converge to normal distributions, this shows that the asymptotic optimal codelength for the symmetric Tardos scheme is $\ell = (\frac{\pi^2}{2} + o(1))c^2 \ln(n/\eps_1)$. We show in Section \ref{sec:Asymptotics} that for asymptotically large $c$, we can actually prove soundness and completeness for this optimal codelength, without any added assumptions. In the asymptotic case of $c \to \infty$, our construction gives the following parameters:
\begin{align*}
d_{\ell} \to \frac{\pi^2}{2}, \quad d_z \to \pi, \quad d_{\delta} \to \infty, \quad d_{\alpha} \to \frac{\pi}{2}, \quad r \to \frac{1}{2}, \quad s \to \infty, \quad g \to \frac{2}{\pi}.
\end{align*}
Similar to the asymmetric case, we thus get a factor $2$ improvement over \v{S}kori\'{c} et al.'s best provable asymptotic codelength, and we achieve the asymptotic optimal codelength which \v{S}kori\'{c} et al.\ only proved with the added assumption that the scores behave like normal distributions. This also improves upon results from Nuida et al.\ in \cite{nuida09}, who showed that with certain discrete distribution functions $F$, one can prove soundness and completeness for $\ell \approx 5.35 c^2 \ln(n/\eps_1)$ for large $c$. With our construction, we show a codelength of $\ell \approx 4.93 c^2 \ln(n/\eps_1)$ is provably secure for large $c$.

\section{Soundness}
\label{sec:Soundness}

Here we will prove Theorem \ref{thm:Soundness}, i.e. prove the soundness property from Definition \ref{def:Secureness}, under the assumptions \eqref{req1} and \eqref{req2}. We will closely follow the proof of soundness of Blayer and Tassa of \cite[Theorem 1.1]{blayer08}. We will first prove an upper bound on $\expn\left[e^{\alpha S_j}\right]$, with $\alpha = 1/(d_{\alpha} c)$ and using only \eqref{req1}, and then use this result together with \eqref{req2} to prove upper bounds on $\pr[j \in \sigma(\vec{y})]$ for innocent users $j$, and $\pr[\sigma(\rho(X)) \not\subseteq C]$.

\begin{lemma} \label{lem:Soundness}
Let $d_{\alpha}$ and $r$ be positive constants, with $r > \frac{1}{2}$, such that $d_{\delta}, d_{\alpha}$ and $r$ satisfy Equation \eqref{req1}. Let $j$ be an innocent user, and let $S_j$ be the user's score in the Tardos scheme from Subsection \ref{sub:Construction-Tardos}. Let $\alpha = 1/(d_{\alpha} c)$. Then
\begin{align}
\expn_{\vec{y},X,\vec{p}}\left[e^{\alpha S_j}\right] \leq e^{+r \alpha^2 \ell}. \label{eq:lem1}
\end{align}
\end{lemma}

\begin{proof}
First we fill in $S_j = \sum_{i = 1}^{\ell} S_{ji}$ and use that $S_j$ does not depend on $X_{j'i}$ for $j' \neq j$ to get
\begin{align*}
\expn_{\vec{y},X,\vec{p}}\left[e^{\alpha S_j}\right] = \expn_{\vec{y},\vec{X}_j,\vec{p}}\left[\prod_{i=1}^{\ell} e^{\alpha S_{ji}}\right] = \prod_{i=1}^{\ell} \expn_{y_i,X_{ji},p_i}\left[e^{\alpha S_{ji}}\right].
\end{align*}
Since $S_{ji} < \sqrt{1/\delta} = \sqrt{d_{\delta} c}$ it follows that $\alpha S_{ji} < \sqrt{d_{\delta}}/(d_{\alpha}\sqrt{c})$. From \eqref{req1} we know that $\sqrt{d_{\delta}}/(d_{\alpha}\sqrt{c}) \leq h(r)$ for our choice of $r$, hence $\alpha S_{ji} < h(r)$. From the definition of $h$ we know that $e^x \leq 1 + x + r x^2$ exactly when $x \leq h(r)$. Using this with $x = \alpha S_{ji}$ we get
\begin{align*}
\expn\left[e^{\alpha S_{ji}}\right] \leq \expn\left[1 + \alpha S_{ji} + r (\alpha S_{ji})^2\right] = 1 + \alpha \expn[S_{ji}] + r \alpha^2 \expn[S_{ji}^2].
\end{align*}
We can easily calculate $\expn[S_{ji}]$ and $\expn[S_{ji}^2]$, as $y_i$ and $X_{ji}$ are independent for innocent users $j$. As in \cite[Lemmas 2 and 3]{skoric08}, we obtain
\begin{align}
\expn[S_{ji}] = 0, \quad \expn[S_{ji}^2] = 1. \label{sound1}
\end{align}
So it follows that $\expn\left[e^{\alpha S_{ji}}\right] \leq 1 + r \alpha^2 \leq e^{r \alpha^2}$, and $\expn_{\vec{y},X,\vec{p}}\left[e^{\alpha S_j}\right] \leq e^{r \alpha^2 \ell}$, which was to be proven.
\end{proof}

\begin{proof}[Theorem \ref{thm:Soundness}]
We prove that the probability of accusing any particular innocent user is at most $\eps_1/n$. Since there are at most $n$ innocent users, the probability of not accusing any innocent users is then at least $(1 - \eps_1/n)^n \geq 1 - \eps_1$, which then proves the scheme is $\eps_1$-sound.

Since a user is accused if and only if his score $S_j$ exceeds $Z$, we need to prove that $\pr[S_j > Z] \leq \eps_1/n$ for innocent users $j$. First of all, we write $\alpha = 1/(d_{\alpha} c)$, and we use the Markov inequality and Lemma \ref{lem:Soundness} to obtain
\begin{align*}
\pr[j \in \sigma(\vec{y})] = \pr[S_j > Z] = \pr\left[e^{\alpha S_j} > e^{\alpha Z}\right] \leq e^{-\alpha Z} \expn\left[e^{\alpha S_j}\right] \leq e^{-\alpha Z + r \alpha^2 \ell}.
\end{align*}
Since we want to prove that $\pr[j \in \sigma(\vec{y})] \leq \eps_1/n$, the proof would be complete if $e^{-\alpha Z + r \alpha^2 \ell} \leq e^{-k} \leq \eps_1/n$, i.e.\ if $-\alpha Z + r \alpha^2 \ell \leq -k$. Filling in $\alpha = 1/(d_{\alpha} c), Z = d_z c k$ and $\ell = d_{\ell} c^2 k$, and dividing both sides by $-k$, we get
\begin{align*}
\frac{d_z}{d_{\alpha}} - \frac{r d_{\ell}}{d_{\alpha}^2} \geq 1.
\end{align*}
This is exactly inequality \eqref{req2}, which was assumed to hold. This completes the proof.
\end{proof}

Compared to the original proof in \cite{blayer08}, this proof has barely changed. The only difference is that now the scores are counted for all positions $i$, instead of only those positions where $y_i = 1$. However, since in the proof in \cite{blayer08} this number of positions was bounded by $\ell$, the result remains the same. This explains why we can prove $\eps_1$-soundness with the symmetric score function under the same assumptions \eqref{req1}, \eqref{req2} as in \cite{blayer08}.

\section{Completeness}
\label{sec:Completeness}

For the proof of Theorem \ref{thm:Completeness}, we will again closely follow the proof of Blayer and Tassa of \cite[Theorem 1.2]{blayer08}, and make changes where necessary to incorporate the symbol-symmetric score function. We first give a Lemma to bound the expectation value of $\expn_{\vec{y},X,\vec{p}}\left[e^{-\beta S}\right]$ with $\beta = s\sqrt{\delta}/c$ and $S = \sum_{j \in C} S_j$, and then use this Lemma to prove completeness.

\begin{lemma}
\label{lem:Completeness}
Let $s$ and $g$ be positive constants such that $d_{\delta}, s$ and $g$ satisfy \eqref{req3}. Let $\beta = s\sqrt{\delta}/c$, let $C$ be a coalition of size $c$, and let $S = \sum_{j \in C} S_j$ be their total coalition score in the Tardos scheme from Subsection \ref{sub:Construction-Tardos}. Then
\begin{align}
\expn_{\vec{y},X,\vec{p}}\left[e^{-\beta S}\right] \leq e^{-g \beta \ell}. \label{eq:lem2}
\end{align}
\end{lemma}

The proof of Lemma \ref{lem:Completeness} is quite lengthy and can be found in Appendix \ref{sec:LemmaCompleteness}. Using this Lemma we can easily prove Theorem \ref{thm:Completeness}.

\begin{proof}[Theorem \ref{thm:Completeness}]
We will prove that for a coalition of size $c$, with probability at least $1 - \eps_2$ the algorithm will accuse at least one of the colluders. Note that if no colluders are accused, then the score of each colluder is below $Z$. Hence if the total coalition score $S$ exceeds $cZ$, then at least one of the pirates is accused. So to prove $\eps_2$-soundness, it suffices to prove that $\pr[S < cZ] \leq \eps_2$.

We first use the Markov inequality and Lemma \ref{lem:Completeness} with $\beta = s\sqrt{\delta}/c > 0$ to get
\begin{align*}
\pr[\sigma(\vec{y}) \cap C = \emptyset] \leq \pr[S < cZ] = \pr\left[e^{-\beta S} > e^{-\beta c Z}\right] \leq e^{\beta c Z} \expn_{\vec{y},X,\vec{p}}\left[e^{-\beta S}\right] \leq e^{\beta c Z - g \beta \ell}.
\end{align*}
Since we want to prove that $\pr[S < cZ] \leq e^{-\eta k} \leq (\eps_1/n)^{\eta} = \eps_2$, the proof would be complete if $\beta c Z - g \beta \ell \leq -\eta k$. Filling in $\beta = s \sqrt{\delta}/c, \ell = d_{\ell} c^2 k, Z = d_z c k, \delta = 1/(d_{\delta} c)$ and writing out both sides, we get
\begin{align*}
g d_{\ell} - d_z \geq \eta \sqrt{\frac{d_{\delta}}{s^2 c}}.
\end{align*}
This is exactly inequality \eqref{req4}, which was assumed to hold. This completes the proof.
\end{proof}

Compared to \cite{blayer08}, we see that instead of using \eqref{btreq3}, we now need that inequality \eqref{req3} holds. Comparing these two inequalities, we see that a term $\frac{1}{\pi}$ has changed to a $\frac{2}{\pi}$, and a term $\frac{2}{d_{\delta} \pi}$ has changed to a $\frac{4}{d_{\delta} \pi}$. The most important change is the $\frac{1}{\pi}$ changing to a $\frac{2}{\pi}$, since that term is the most dominant factor (and the only positive term) on the left hand side of \eqref{req3}. By increasing this by a factor $2$, we get that $g \leq \frac{2}{\pi}$ instead of $g \leq \frac{1}{\pi}$. Especially for large $c$, this will play an important role, and it will basically be the reason why the required codelength can then be reduced by a factor $4$, compared to Blayer and Tassa's analysis for the asymmetric scheme.

While the other change (the $\frac{2}{d_{\delta} \pi}$ changing to $\frac{4}{d_{\delta} \pi}$) does not have a big impact on the optimal choice of parameters for large $c$, this change does influence the required codelength for smaller $c$. Because of this change, we now subtract more from the left hand side of \eqref{req3}, so that the value of $g$ is bounded sharper from above. This means that for finite $c$ we cannot reduce the codelength of Blayer and Tassa by a factor $4$, but only by a factor slightly less than $4$.

Finally, after using \eqref{req3} in the proof above, the analysis remained the same as in \cite{blayer08}. So under the same assumption \eqref{req4} as in \cite{blayer08}, we could also complete the proof for the symmetric Tardos scheme.

\section{Optimization}
\label{sec:Optimization}

Similar to the analysis done by Blayer and Tassa in \cite[Section 2.4]{blayer08}, we also investigate the optimal choice of parameters such that all requirements are satisfied, and $d_{\ell}$ is minimized. As only one of the inequalities has changed, and it changed only on two positions, the formulas for the optimal values of $d_{\delta}, d_{\alpha}, d_z, d_{\ell}$ in the following Theorem are almost the same as in \cite[Section 2.4.5]{blayer08}. We do not give a proof here, as it would be nearly identical to the analysis done in \cite[Section 2.4]{blayer08}.

\begin{theorem} \label{thm:Optimization}
Let $\eta, c$ be given, and let $r,s,g$ be fixed, satisfying $r \in (\frac{1}{2},\infty), s \in (0, \infty), g \in (0, \frac{2}{\pi})$. Then the optimal choice of $d_{\delta},d_{\alpha},d_z,d_{\ell}$, minimizing $d_{\ell}$ and satisfying conditions \eqref{req1},\eqref{req2},\eqref{req3},\eqref{req4}, is given by:
\begin{align*}
\hat{d}_{\delta} & = \left(\frac{1}{\frac{4}{\pi} - 2g}\left(\sqrt{\frac{(h^{-1}(s) s)^2}{c} + \frac{16}{\pi}\left(\frac{2}{\pi} - g\right)} + \frac{h^{-1}(s) s}{\sqrt{c}}\right)\right)^2, \tag{O1} \label{opt1} \\
\hat{d}_{\alpha} & = \max\left(\frac{\sqrt{\hat{d}_{\delta}}}{h(r) \sqrt{c}}, \frac{r}{g} + \sqrt{\left(\frac{r}{g}\right)^2 + \frac{r}{g} \eta \sqrt{\frac{\hat{d}_{\delta}}{s^2 c}}}\right), \tag{O2} \label{opt2} \\
\hat{d}_z & = \frac{g \hat{d}_{\alpha}^2 + r \eta \sqrt{\frac{\hat{d}_{\delta}}{s^2 c}}}{g \hat{d}_{\alpha} - r}, \tag{O3} \label{opt3} \\
\end{align*}
\begin{align*}
\hat{d}_{\ell} & = \frac{\eta \sqrt{\frac{\hat{d}_{\delta}}{s^2 c}} + \hat{d}_z}{g}. \tag{O4} \label{opt4}
\end{align*}
\end{theorem}

So to find the optimal septuple $(\hat{r},\hat{s},\hat{g},\hat{d}_{\delta},\hat{d}_{\alpha},\hat{d}_z,\hat{d}_{\ell})$ for given $c,\eta$, satisfying all requirements and minimizing $\hat{d}_{\ell}$, one only has to find the triple $(r,s,g)$ with $r \in (\frac{1}{2}, \infty), s \in (0, \infty)$ and $g \in (0, \frac{2}{\pi})$ that minimizes the right hand side of \eqref{opt4}.

\paragraph{Example} An optimal solution to \eqref{req1},\eqref{req2},\eqref{req3},\eqref{req4} for $c \geq 2$ and $\eta = 1$, minimizing $d_{\ell}$, is given by
\begin{align*}
d_{\ell} = 23.79, \quad d_z = 8.06, \quad d_{\delta} = 28.31, \quad d_{\alpha} = 4.58, \quad r = 0.67, \quad s = 1.07, \quad g = 0.49.
\end{align*}
This means that with these constants, we can prove soundness and completeness for all $c \geq 2$ and $\eta \leq 1$, with a codelength of $\ell = 23.79 c^2 \ln(n/\eps_1)$. Compared to the original Tardos scheme, which had a codelength of $\ell = 100 c^2 \left\lceil \ln(n/\eps_1) \right\rceil$, this gives an improvement of a factor more than $4$. Furthermore we can prove that this scheme is $\eps_1$-sound and $\eps_2$-complete for any value of $c \geq 2$ and $\eta \leq 1$, while Tardos' original proof only works for $c \geq 2$ and $\eta \leq \sqrt{c}/4$.

\paragraph{Example} In practice, one usually has $\eta \ll 1$ instead of $\eta = 1$. For example, it could be that $\eps_2 = 1/2$ is sufficient, while $\eps_1 = 10^{-3}$ is desired and there are $n = 10^6$ users, so that $\eta \approx 0.033$. Then the optimizations give us $d_{\ell} \approx 10.89$ for $c = 2$. So with this larger value of $\eps_2$, a codelength of $\ell = 10.89 c^2 \ln(n/\eps_1)$ is sufficient to prove the soundness and completeness properties for any $c \geq 2$. This is then already a factor more than $9$ improvement compared to the original Tardos scheme.

If we let $c$ increase in inequalities \eqref{opt1},\eqref{opt2},\eqref{opt3},\eqref{opt4}, i.e. if we only want provable soundness and completeness for $c \geq c_0$ for some $c_0 > 2$, then one can easily see that the inequalities become weaker and 

\begin{figure}[h]
\includegraphics[width=\textwidth]{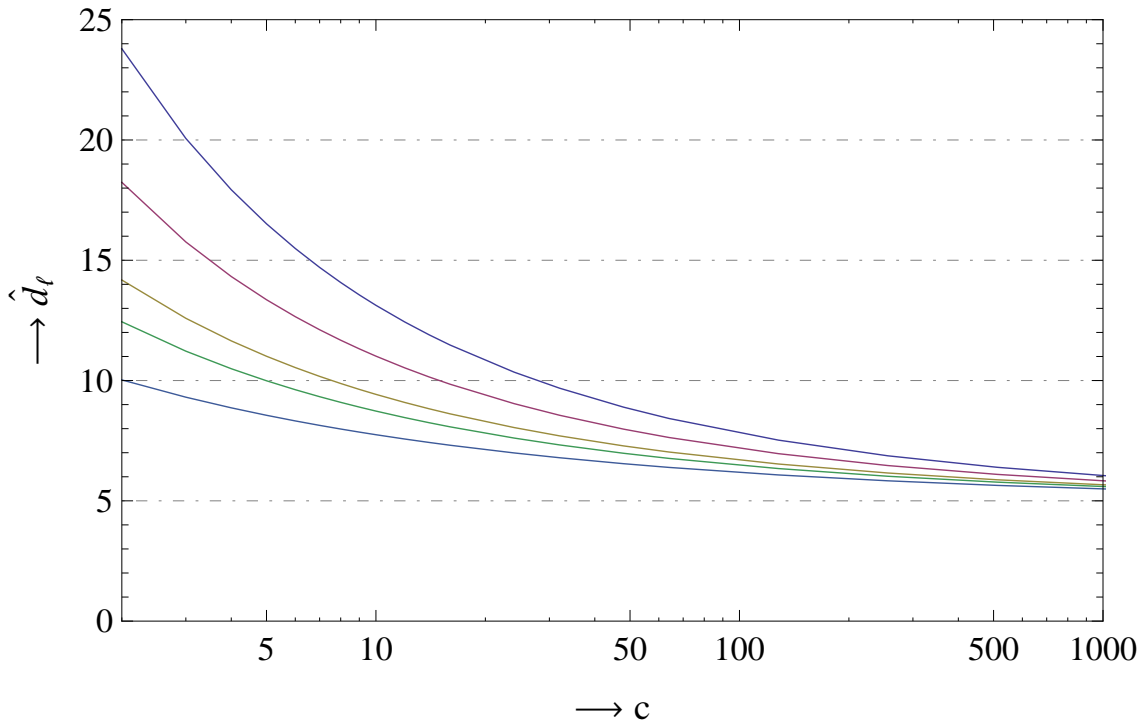}
\caption{Optimal values of $d_{\ell}$, for several values of $c$ between $2$ and $1000$. The different lines correspond to the cases $\eta = 1, 0.5, 0.2, 0.1, 0.01$ respectively, where higher values of $\eta$ correspond to higher values of $\hat{d}_{\ell}$.}
\label{fig:Fig1}
\end{figure}

\noindent
an even shorter codelength can be achieved. Figure \ref{fig:Fig1} shows the optimal values of $d_{\ell}$ against different values of $c$, for several values of $\eta$. One can see that for large $c$, a codelength of $\ell < 6 c^2 \ln(n/\eps_1)$ can be sufficient. In the next Section, we will see that for large $c$, the optimal values of $d_{\ell}$ will converge to $\frac{\pi^2}{2} \approx 4.93$.

\section{Asymptotics}
\label{sec:Asymptotics}

Here we show that with the symmetric Tardos construction, for $c \rightarrow \infty$ we can prove soundness and completeness for $d_{\ell} = \frac{\pi^2}{2} + O\left(c^{-1/3}\right)$. We calculate the optimal first order error term explicitly, and also show explicitly the dependence on $\eta$, as the choice of $\eta$ may depend on the particular application. Note that at least $\eta \leq 1$, but it may be considerably smaller and it may depend on $c$ as well.

\begin{theorem} \label{thm:FirstOrder}
Let $\gamma = \left(\frac{2}{3\pi} \right)^{2/3} \approx 0.35577$. The optimal asymptotic (for $c \to \infty$) value for $d_{\ell}$, and the accompanying values for $d_z, d_{\delta}$, are
\begin{align}
  d_{\ell} &= \frac{\pi^2}{2} \left(1 + \left(3 \gamma + 18 \gamma \frac{\eta}{\log c} (1 + o(1)) \right) c^{-1/3} \right), \\
  d_z &= \pi \left(1 + \left(\frac{5}{2} \gamma + 6 \gamma \frac{\eta}{\log c} (1 + o(1)) \right) c^{-1/3} \right), \\
  d_{\delta} &= \frac{4}{\gamma} \left(1 - 3 \frac{\eta}{\log c}(1 + o(1)) \right) c^{1/3},
\end{align}
and the choices for $g, r, s$ leading to them are given by
\begin{align}
  g &= \frac{2}{\pi} \left(1 - \left(\frac{1}{2}\gamma + 3 \gamma \frac{\eta}{\log c} (1 + o(1)) \right) c^{-1/3} \right), \\
  r &= \frac{1}{2} \left(1 + \left(2 \gamma - 6 \gamma \frac{\eta}{\log c} (1 + o(1)) \right) c^{-1/3} \right), \\
  s &= \log \left(\frac{24}{\pi^2\gamma} \frac{\eta}{\log c} (1 + o(1)) c^{1/3} \right).
\end{align}
\end{theorem}

We have optimized for $d_{\ell}$, and one could get slightly better error terms for $d_z$ or $d_{\delta}$. For example, optimizing for $d_z$ yields an optimal value of $\pi(1 + (3 \gamma' + 9 \gamma' \frac{\eta}{\log c} (1 + o(1))) c^{-1/3})$, for a suboptimal $d_{\ell}$ of $\frac{\pi^2}{2}(1 + (4 \gamma' + 15 \gamma' \frac{\eta}{\log c} (1 + o(1))) c^{-1/3})$, where $\gamma' = 2^{-1/3} \gamma$.

It is remarkable that the error terms for $d_{\ell}$ and $d_z$ scale with $c^{-1/3}$, while \v{S}kori\'{c} et al.\ found error terms scaling with $c^{-1/2}$. It turns out that in \cite{skoric08} an error term in $\widetilde{\mu}$ was not taken into account, and if one does do, their analysis for the binary case will also yield error terms scaling with $c^{-1/3}$. Also note that $d_{\delta}$ (related to the cutoff) scales with $c^{1/3}$, i.e.\ the cutoff $\frac{1}{d_{\delta} c}$ scales with $c^{-4/3}$ rather than with $c^{-1}$ as one might have guessed.

An immediate consequence of Theorem \ref{thm:FirstOrder} is the following result, which shows that asymptotically we will achieve codelengths of $\ell \approx 4.93 c^2 \ln(n/\eps_1)$, i.e. codelengths that are about $4.93\%$ of Tardos' original codelengths.

\begin{corollary} \label{cor:Asymptotics}
For $c \to \infty$ the above construction gives an $\eps_1$-sound and $\eps_2$-complete scheme with parameters
\begin{align*}
\ell \to \frac{\pi^2}{2} c^2 \ln(n/\eps_1), \quad Z \to \pi c \ln(n/\eps_1), \quad \delta \to \frac{\gamma}{4} c^{-4/3}.
\end{align*}
\end{corollary}

This proves that our analysis is asymptotically tight, since for large $c$ we achieve the optimal codelength of $\ell = (\frac{\pi^2}{2} + o(1)) c^2 \ln(n/\eps_1)$.

\paragraph{Remark} In the proof of Theorem \ref{thm:FirstOrder}, we use that $r$ can be taken in the neighborhood of $\frac{1}{2}$ to get the final result, $d_{\ell} = \frac{\pi^2}{2} + O(c^{-1/3})$. In \cite{skoric08} however, no such variable $r$ was used, as it was simply fixed at $1$. If they had taken $r$ as a parameter in their analysis and had taken it close to $\frac{1}{2}$ in the asymptotic case, then they would have obtained the same asymptotic results as we did above, but still with different first order terms.

\section*{Acknowledgment}

The authors would like to thank Boris \v{S}kori\'{c}, Jeroen Doumen and Peter Roelse for useful discussions and valuable comments.

\appendix
\small

\section{Integral codelengths}
\label{sec:IntegralCodelengths}

One detail we have not taken care of and which is often "{}swept under the carpet"{}\ in other literature, is that the codelength $\ell$ by definition has to be integral. In the construction of the Tardos scheme however, we said we take $\ell = d_{\ell} c^2 \ln(n/\eps_1)$, while $\ln(n/\eps_1)$ and $d_{\ell}$ may not be integral. To solve the problem of non-integral codelengths, Tardos rounded up $\ln(n/\eps_1)$ and took $d_{\ell} = 100$ in his original scheme. Blayer and Tassa also rounded up $\ln(n/\eps_1)$ and took $d_{\ell} = 85$, presumably also to guarantee that $\ell$ is integral\footnote{Numerical optimizations show that even a parameter set with $d_{\ell} \approx 81.25$ exists that satisfies all requirements of Blayer and Tassa.}. However, rounding up $d_{\ell}$ and $\ln(n/\eps_1)$ could drastically increase the codelength. For example, suppose $n = 10^6, \ \eps_1 = \eps_2 = 0.01$, and $c = 25$. Then $\eta = 0.25$ and $\ln(n/\eps_1) \approx 18.42$, and numerical optimizations give $d_{\ell} \approx 8.18$. Without rounding we would get a codelength of $\ell \approx 94155$, while with rounding we get $\ell' = 106875$. So then the codelength $\ell'$ is more than $13.5\%$ higher than $\ell$, only because we rounded up both $\ln(n/\eps_1)$ and $d_{\ell}$.

Instead of rounding up inbetween, rounding up the entire codelength to $\ell' = \lceil d_{\ell} c^2 \ln(n/\eps_1) \rceil$ makes more sense. The codelength is then increased by less than $1$ symbol, so we hardly notice the difference in the codelength. However, the proofs we give in Section \ref{sec:Soundness} and \ref{sec:Completeness} are based on $\ell = d_{\ell} c^2 \ln(n/\eps_1)$, which corresponds to using $d_{\ell} = \ell / (c^2 \ln(n/\eps_1))$. If we take $\ell' = \lceil \ell \rceil$, then we get $d_{\ell}' = \lceil \ell \rceil / (c^2 \ln(n/\eps_1)) > d_{\ell}$ (for $\ell \notin \mathbb{N}$), so that with the same parameters $Z$ and $\delta$ we may not be able to prove soundness and completeness anymore. In particular, equation \eqref{req2} might not be satisfied if $d_{\ell}$ is increased, since \eqref{req2} implies that $4 r d_{\ell} \leq d_z^2$. Increasing the left hand side may violate this bound, if we do not also increase $d_z$.

The following Theorem takes care of this minor problem, by showing that if we can find a solution to \eqref{req1}, \eqref{req2}, \eqref{req3}, \eqref{req4} with a fractional codelength $\ell$, then we can also find a solution to these inequalities with the integral codelength $\lceil \ell \rceil$. In particular, we show which scheme parameters $\ell, \ Z$ and $\delta$ one could take to achieve this result.

\begin{theorem}
Let the Tardos scheme be constructed as in Subsection \ref{sub:Construction-Tardos}, and let $(d_{\ell}, d_z, d_{\delta}, d_{\alpha}, r, s, g)$ be a septuple satisfying conditions \eqref{req1}, \eqref{req2}, \eqref{req3}, \eqref{req4} giving scheme parameters $\ell_0 = d_{\ell} c^2 \ln(n/\eps_1), Z_0 = d_z c \ln(n/\eps_1)$ and $\delta_0 = 1/(d_{\delta} c)$. Then the Tardos scheme from Subsection \ref{sub:Construction-Tardos} with parameters
\begin{align}
\ell = \lceil \ell_0 \rceil, \quad Z = Z_0 + \frac{g}{c} \left(\lceil \ell_0 \rceil - \ell_0\right), \quad \delta = \delta_0, \label{intell}
\end{align}
is $\eps_1$-sound and $\eps_2$-complete.
\end{theorem}

\begin{proof}
Let us write $\omega = d_{\ell} (\lceil \ell_0 \rceil - \ell_0)/\ell_0$. We prove that if the equations hold for $(d_{\ell}, d_z, d_{\delta}, d_{\alpha}, r, s, g)$, then they also hold for $(d_{\ell}', d_z', d_{\delta}, d_{\alpha}', r, s, g)$, where $d_{\ell}' = d_{\ell} + \omega, d_z' = d_z + g \omega, d_{\alpha}' = (d_z' + \sqrt{(d_z')^2 - 4 r d_{\ell}'})/2$. Since for this set of parameters we get $\ell, Z$ and $\delta$ as in \eqref{intell}, the result then follows.

First note that since $d_{\delta}, s$ and $g$ did not change, both sides of inequality \eqref{req3} remain the same and this inequality is still satisfied. For inequality \eqref{req4}, note that both sides also remained the same, since $g d_{\ell}' - d_z' = g (d_{\ell} + \omega) - (d_z + g \omega) = g d_{\ell} - d_z$. For \eqref{req2}, we rewrite this inequality as a quadratic inequality in $d_{\alpha}'$:
\begin{align}
(d_{\alpha}')^2 + (-d_z') d_{\alpha}' + r d_{\ell}' \leq 0. \label{integral1}
\end{align}
This inequality is satisfied if and only if $d_{\alpha}'$ lies between the two roots of $x^2 + (-d_z') x + r d_{\ell}' = 0$, which therefore must exist. These roots exist if and only if $(d_z')^2 - 4 r d_{\ell}' \geq 0$. Since we know that $d_z^2 - 4 r d_{\ell} \geq 0$ the inequality follows if
\begin{align*}
(d_z')^2 - 4 r d_{\ell}' = (d_z^2 - 4 r d_{\ell}) + (2 g d_z + g^2 \omega^2 - 4 r) \geq d_z^2 - 4 r d_{\ell} \geq 0.
\end{align*}
From \eqref{req2} and \eqref{req4} we know that $g (d_z^2) \geq g (4 r d_{\ell}) \geq 4 r d_z$, i.e. $g d_z \geq 4 r$. So it follows that $2 g d_z + g^2 \omega^2 \geq 4 r$, which proves the second inequality. The third inequality then follows from \eqref{req2}.

Finally for \eqref{req1}, we prove that $d_{\alpha}' \geq d_{\alpha}$, while the right hand side remains the same, so that this inequality is still satisfied. Note that $d_{\alpha}$ is also at most the largest root of \eqref{integral1}, so $d_{\alpha}' - d_{\alpha}$ is bounded by
\begin{align*}
d_{\alpha}' - d_{\alpha} \geq \frac{d_z' + \sqrt{(d_z')^2 - 4 r d_{\ell}'}}{2} - \frac{d_z + \sqrt{d_z^2 - 4 r d_{\ell}}}{2} \geq \frac{g \omega}{2} \geq 0.
\end{align*}
Here the second inequality follows from earlier calculations that $(d_z')^2 - 4 r d_{\ell}' \geq d_z^2 - 4 r d_{\ell}$. So this choice of $d_{\alpha}'$ is at least as high as $d_{\alpha}$, so inequality \eqref{req1} is satisfied. This completes the proof.
\end{proof}

\section{Proof of Lemma \ref{lem:Completeness}}
\label{sec:LemmaCompleteness}

For proving Lemma \ref{lem:Completeness} we will again closely follow the analysis of Blayer and Tassa, and make changes where necessary.

First, we write the total accusation sum of all colluders together as follows:
\begin{align*}
S = \sum_{i=1}^{\ell} \sum_{j \in C}^c S_{ji} = \sum_{i=1}^{\ell} y_i\left(x_i q_i - \frac{c - x_i}{q_i}\right) + \sum_{i=1}^{\ell} (1 - y_i)\left(\frac{c - x_i}{q_i} - x_i q_i\right).
\end{align*}
Here $x_i$ is the number of ones on the $i$th positions of all colluders, $y_i$ is the output symbol of the pirates on position $i$, and we introduced the notation $q_i = \sqrt{(1 - p_i)/p_i}$. Following the analysis from e.g.\ Blayer and Tassa, and Tardos, but using that $S_i = (1 - y_i)\left(\frac{c - x_i}{q_i} - x_i q_i\right)$ for positions $i$ where $y_i = 0$ (instead of $S_i = 0$, as with the asymmetric score function), we can bound the expectation value by
\begin{align}
\expn_{\vec{y},X,\vec{p}}\left[e^{-\beta S}\right] \leq \left(\sum_{x=0}^c \binom{c}{x} M_x\right)^{\ell}, \label{comp1}
\end{align}
where
\begin{align*}
M_x = \begin{cases}
  E_{0,x} & \text{if $x = 0$}, \\
  E_{1,x} & \text{if $x = c$}, \\
  \max(E_{0,x}, E_{1,x}) & \text{otherwise},
  \end{cases}
\end{align*}
and, for some random variable $p$ distributed according to $F$,
\begin{align*}
E_{0,x} & = \expn_p\left[p^x (1 - p)^{c - x} e^{-\beta\left(\frac{c - x}{q} - x q\right)}\right], \\
E_{1,x} & = \expn_p\left[p^x (1 - p)^{c - x} e^{-\beta\left(x q - \frac{c - x}{q}\right)}\right].
\end{align*}

Now, using $\beta = s \sqrt{\delta}/c$, we bound the exponents in $E_{0,x}$ and $E_{1,x}$ as follows.
\begin{align*}
-s = \frac{-\beta c}{\sqrt{\delta}} \leq -\beta c q \leq -\beta \left(x q - \frac{c - x}{q}\right) \leq \frac{\beta c}{q} \leq \frac{\beta c}{\sqrt{\delta}} = s.
\end{align*}
So $|\beta (x q - (c - x)/q)| \leq s$ for our choice of $\beta$. So we can use the inequality $e^w \leq 1 + w + h^{-1}(s) w^2$ which holds for all $w \leq s$, with $w = \pm \beta (x q - (c - x)/q)$, to obtain
\begin{align*}
E_{0,x} & \leq \expn_p\left[p^x (1 - p)^{c - x} \left(1 + \beta \left(x q - \frac{c - x}{q}\right) + h^{-1}(s) \beta^2 \left(x q - \frac{c - x}{q}\right)^2\right)\right], \\
E_{1,x} & \leq \expn_p\left[p^x (1 - p)^{c - x} \left(1 - \beta \left(x q - \frac{c - x}{q}\right) + h^{-1}(s) \beta^2 \left(x q - \frac{c - x}{q}\right)^2\right)\right].
\end{align*}
Introducing more notation, this can be rewritten to
\begin{align*}
E_{0,x} & \leq F_{0,x} + \beta F_{1,x} + h^{-1}(s) \beta^2 F_{2,x}, \\
E_{1,x} & \leq F_{0,x} - \beta F_{1,x} + h^{-1}(s) \beta^2 F_{2,x},
\end{align*}
where
\begin{align*}
F_{0,x} & = \expn_p\left[p^x (1 - p)^{c - x}\right], \\
F_{1,x} & = \expn_p\left[p^x (1 - p)^{c - x} \left(x q - \frac{c - x}{q}\right)\right], \\
F_{2,x} & = \expn_p\left[p^x (1 - p)^{c - x} \left(x q - \frac{c - x}{q}\right)^2\right].
\end{align*}
We first calculate $F_{1,x}$ explicitly. Writing out the expectation value and using the definition of $f(p)$ from \eqref{dist}, we get
\begin{align*}
F_{1,x} = \frac{1}{\pi - 4\delta'} \int_{\delta}^{1 - \delta} p^x (1 - p)^{c - x} \left(\frac{x}{p} - \frac{c - x}{1 - p}\right) dp
\end{align*}
The primitive of the integrand is given by $I(p) = p^x (1 - p)^{c - x}$, so we get
\begin{align}
F_{1,x} & = \frac{I\left(1 - \delta\right) - I(\delta)}{\pi - 4\delta'} = \frac{(1 - \delta)^x \delta^{c - x} - \delta^x (1 - \delta)^{c - x}}{\pi - 4\delta'}. \label{comp3}
\end{align}
For $0 < x < c$, we bound $F_{1,x}$ from above and below as
\begin{align*}
\frac{-\delta^x (1 - \delta)^{c - x}}{\pi - 4\delta'} \leq F_{1,x} \leq \frac{(1 - \delta)^x \delta^{c - x}}{\pi - 4\delta'}.
\end{align*}
Using these bounds for $M_x$, with $0 < x < c$, we get
\begin{align*}
M_x & \leq F_{0,x} + \beta\frac{\max(\delta^x (1 - \delta)^{c - x}, (1 - \delta)^x \delta^{c - x})}{\pi - 4\delta'} + h^{-1}(s) \beta^2 F_{2,x}.
\end{align*}
Since $\delta < 1 - \delta$, the maximum of the two terms is the first term when $0 < x \leq c/2$, and it is the second term when $c/2 < x < c$. Note that this bound is different from the one of Blayer and Tassa, since in their analysis they do not have this maximum over two terms, but just the first of these two terms. We cannot prove the same upper bound as Blayer and Tassa, and therefore our bound for $M_x, 0 < x < c$, is slightly weaker than Blayer and Tassa's.

For the positions where the marking assumption applies, i.e.\ $x = 0$ and $x = c$, we do not use the bounds on $F_{1,x}$, but use the exact formula from \eqref{comp3} to obtain
\begin{align*}
M_0 & \leq F_{0,0} - \beta \frac{(1 - \delta)^c - \delta^c}{\pi - 4\delta'} + h^{-1}(s) \beta^2 F_{2,0}, \\
M_c & \leq F_{0,c} - \beta \frac{(1 - \delta)^c - \delta^c}{\pi - 4\delta'} + h^{-1}(s) \beta^2 F_{2,c}.
\end{align*}
The value of $M_c$ is the same as that of Blayer and Tassa, but whereas Blayer and Tassa had $M_0 = F_0$, we get a lower upper bound on $M_0$. This is essentially the reason why with the symmetric score function we get shorter codelengths than Blayer and Tassa.

Substituting the bounds on $M_x$ in the summation over $M_x$ from \eqref{comp1} gives us
\begin{align}
\sum_{x=0}^c \binom{c}{x} M_x & \leq M_0 + M_c + \sum_{x=1}^{c-1} \binom{c}{x} M_x \notag \\
  & \leq \left(F_{0,0} - \beta \frac{(1 - \delta)^c - \delta^c}{\pi - 4\delta'} + h^{-1}(s) \beta^2 F_{2,0}\right) \notag \\
	& + \left(F_{0,c} - \beta \frac{(1 - \delta)^c - \delta^c}{\pi - 4\delta'} + h^{-1}(s) \beta^2 F_{2,c}\right) \notag \\
  & + \sum_{x=1}^{\left\lfloor c/2 \right\rfloor} \binom{c}{x} \left(F_{0,x} + \beta\frac{\delta^x (1 - \delta)^{c - x}}{\pi - 4\delta'} + h^{-1}(s) \beta^2 F_{2,x}\right) \notag  \\
	& + \sum_{x=\left\lfloor c/2 \right\rfloor + 1}^{c-1} \binom{c}{x} \left(F_{0,x} + \beta\frac{(1 - \delta)^x \delta^{c - x}}{\pi - 4\delta'} + h^{-1}(s) \beta^2 F_{2,x}\right). \notag
\end{align}
Gathering all terms with $F_{0,x}$ and $F_{2,x}$, and using the substitution $x' = c - x$ for the summation over $\left\lceil c/2 \right\rceil - 1$ terms, we get
\begin{align}
\sum_{x=0}^c \binom{c}{x} M_x & \leq \sum_{x=0}^c \binom{c}{x} F_{0,x} - \beta \frac{2 (1 - \delta)^c}{\pi - 4\delta'} + h^{-1}(s) \beta^2 \sum_{x=0}^c \binom{c}{x} F_{2,x} \notag \\
  & + \frac{\beta}{\pi - 4\delta'} \left(\delta^c + \sum_{x=1}^{\left\lfloor c/2 \right\rfloor} \binom{c}{x} \delta^x (1 - \delta)^{c - x}\right) \notag \\
	& + \frac{\beta}{\pi - 4\delta'} \left(\delta^c + \sum_{x'=1}^{\left\lceil c/2 \right\rceil - 1} \binom{c}{x'} \delta^{x'} (1 - \delta)^{c - x'}\right). \label{comp2}
\end{align}
For the summation over $F_{2,x}$, let us define a sequence of random variables $\{T_i\}_{i=1}^c$ according to $T_i = q$ with probability $p$ and $T_i = -1/q$ with probability $1 - p$. Similar to the inequalities from \eqref{sound1}, we get that $\expn_p[T_i] = 0$ and $\expn_p[T_i^2] = 1$. Also, since $T_i$ and $T_j$ are independent for $i \neq j$, we have that $\expn_p[T_i T_j] = 0$ for $i \neq j$. Therefore we can write
\begin{align*}
\expn_p\left[\left(\sum_{i=1}^c T_i\right)^2\right] = \sum_{i=1}^c \expn_p\left[T_i^2\right] + \sum_{i \neq j} \expn_p\left[T_i T_j\right] = c.
\end{align*}
But writing out the definition of the expected value, we see that the left hand side is actually the same as the summation over $F_{2,x}$, so that we get
\begin{align*}
\expn_p\left[\left(\sum_{i=1}^c T_i\right)^2\right] = \sum_{x=0}^c \binom{c}{x} p^x (1 - p)^{c-x} \left(x q - \frac{c - x}{q}\right)^2 = \sum_{x=0}^c \binom{c}{x} F_{2,x} = c.
\end{align*}
Also we trivially have that
\begin{align*}
\sum_{x=0}^c \binom{c}{x} F_{0,x} = \sum_{x=0}^c \binom{c}{x} \expn_p\left[p^x (1 - p)^{c - x}\right] = \expn_p\left[\sum_{x=0}^c \binom{c}{x} p^x (1 - p)^{c - x}\right] = 1.
\end{align*}
For the summation over $\lfloor c/2 \rfloor$ terms we use the following upper bound, which then also holds for the summation over $\lceil c/2 \rceil - 1$ terms:
\begin{align*}
\delta^c + \sum_{x=1}^{\left\lfloor c/2 \right\rfloor} \binom{c}{x} \delta^x (1 - \delta)^{c - x} \leq \sum_{x=1}^c \binom{c}{x} \delta^x (1 - \delta)^{c - x} = 1 - (1 - \delta)^c \leq \delta c.
\end{align*}
Note that this first inequality is quite sharp. In most cases $\delta \ll 1 - \delta$, so that the summation is dominated by the terms with low values of $x$. Adding the terms with $\left\lfloor c/2 \right\rfloor < x < c$ (i.e. terms with high powers of $\delta$) to the summation has an almost negligible effect on the value of the summation.

Now applying the previous results to \eqref{comp2}, and using $(1 - \delta)^c \geq 1 - \delta c$, which holds for all $c$, we get
\begin{align*}
\sum_{x=0}^c \binom{c}{x} M_x & \leq 1 - \beta \frac{2 - 4 c \delta}{\pi - 4\delta'} + h^{-1}(s) \beta^2 c.
\end{align*}
We want to prove that, for some $g > 0$,
\begin{align}
\sum_{x=0}^c \binom{c}{x} M_x & \leq 1 - \beta \frac{2 - 4 c \delta}{\pi - 4\delta'} + h^{-1}(s) \beta^2 c \leq 1 - g \beta \leq e^{-g \beta}. \label{comp4}
\end{align}
Filling in $\beta = s \sqrt{\delta}/c$ and $\delta = 1/(d_{\delta} c)$ and writing out the second inequality, this leads to the requirement that
\begin{align*}
\frac{2 - \frac{4}{d_{\delta}}}{\pi} - \frac{h^{-1}(s) s}{\sqrt{d_{\delta} c}} \geq g.
\end{align*}
This is exactly inequality \eqref{req3}, which is assumed to hold. Combining the results from Equations \eqref{comp4} and \eqref{comp2} gives us
\begin{align*}
\expn_{\vec{y},X,\vec{p}}\left[e^{-\beta S}\right] \leq \left(\sum_{x=0}^c \binom{c}{x} M_x\right)^{\ell} \leq e^{-g \beta \ell}.
\end{align*}
This completes the proof.
\qed

\section{Proof of Theorem \ref{thm:FirstOrder}}
\label{sec:AppendixFirstOrder}

We introduce parameters $K_g, K_r, K_s$, a priori depending on $c$, to enable us to write
\begin{align*}
g = \frac{2}{\pi} - K_g c^{-1/3}, \quad h(r) = K_r c^{-1/3}, \quad \frac{1}{sh^{-1}(s)} = K_s c^{-1/3}.
\end{align*}
Clearly $K_g, K_r, K_s$ are positive, and we will assume that $K_g$ and $K_r$ are $O(1)$ for $c \rightarrow \infty$. This assumption will be validated later on. Note that we do not demand this for $K_s$ (and indeed, it will turn out that $K_s \to \infty$).

Note that $r = h^{-1}(K_r c^{-1/3}) = \frac{1}{2} + \frac{1}{6} K_r c^{-1/3} + O(c^{-2/3})$, so that, with for convenience $R = \frac{r}{g}$, we have
\begin{align}
R = \frac{\pi}{4} + \left(\frac{\pi^2}{8} K_g + \frac{\pi}{12} K_r \right) c^{-1/3} + O\left(c^{-2/3}\right). \label{eq:R}
\end{align}
Next, for convenience we put $D = \sqrt{\frac{d_{\delta}}{c}}$, and then we have from (O1) that $D = D_0 c^{-1/3}$, where
\begin{align*}
D_0 = \frac{1}{2 K_g K_s} \left(1 + \sqrt{1 + \frac{16}{\pi} K_g K_s^2}\right).
\end{align*}
Note that $D_0$ is a decreasing function of $K_s$, with limiting value $\frac{2}{\sqrt{\pi}} \frac{1}{\sqrt{K_g}}$.

From (O2) we see that $d_{\alpha} = \max \left\{\frac{D}{h(r)}, x_0 \right\} $, where $x_0 = R + \sqrt{R^2 + R D \frac{\eta}{s}}$. Note that
\begin{align}
x_0 = 2 R + \frac{1}{2} D \frac{\eta}{s} + O\left(c^{-2/3}\right), \label{eq:X}
\end{align}
where we used that $\frac{\eta}{s} = o(1)$. Note that (O3) and (O4) imply $d_{\ell} = \frac{d_{\alpha}^2 + d_{\alpha} D \frac{\eta}{s}}{g (d_{\alpha} - R)}$, and that by $d_{\alpha} > R$ we have $d_{\ell} \geq \frac{2 x_0 + D \frac{\eta}{s}}{g}$, with equality if and only if $d_{\alpha} = x_0$. So in order to minimize $d_{\ell}$ we minimize $\frac{2 x_0 + D \frac{\eta}{s}}{g}$, and show that there is a solution to this with $d_{\alpha} = x_0$, which then must be the optimum. For this optimal solution, by \eqref{eq:R} and \eqref{eq:X} we get
\begin{align}
d_{\ell} = \frac{\pi^2}{2} + L_0 c^{-1/3} + O(c^{-2/3}), \textrm{ where } L_0 = \frac{\pi^3}{2} K_g + \frac{\pi^2}{6} K_r + \pi D_0 \frac{\eta}{s}. \label{eq:L}
\end{align}
To find the main terms in the optimal values for $d_{\ell}, d_z, d_{\delta}$, for the moment we neglect error terms. The fact that $d_{\alpha} = x_0$ implies that $\frac{D}{h(r)} \leq x_0$, and this is asymptotically equivalent to $\frac{D_0}{K_r} \leq \frac{\pi}{2}$. This can be expanded into $1 + \sqrt{1 + \frac{16}{\pi} K_g K_s^2} \leq \pi K_g K_r K_s$, and this leads to $(\pi^3 K_g K_r^2 - 16) K_s \geq 2 \pi^2 K_r$, which actually is two conditions:
\begin{align}
K_g K_r^2 > \frac{16}{\pi^3} = 0.51602\ldots, \quad K_s \geq \frac{2 \pi^2 K_r}{\pi^3 K_g K_r^2 - 16}. \label{eq:I}
\end{align}
This shows that it is impossible to choose both $K_g$ and $K_r$ close to $ 0 $, and that it is certainly possible to choose them $O(1)$ as $c \to \infty$. Note that optimizing $\frac{\eta}{s}$ implies taking $s$ as large as possible, but this means taking $K_s$ as small as possible, which is limited by the above condition. Indeed, in minimizing $L_0$ we would like to minimize $K_g$ and $K_r$, leading to growing $K_s$, while also $s$
preferably keeps growing. We will see that this is possible.

In optimizing $L_0$, to find the main term we also neglect for the moment the term $\pi D_0 \frac{\eta}{s}$, as it also tends to $0$. So we optimize $L_0' = \frac{\pi^3}{2} K_g + \frac{\pi^2}{6} K_r$ under the constraint $K_g K_r^2 > \frac{16}{\pi^3}$. The minimal value for $L_0'$ is reached for $K_g \to \frac{\gamma}{\pi} \approx 0.11325, K_r \to 6 \gamma = 2.1346$, where $\gamma = (\frac{2}{3\pi})^{2/3} \approx 0.35577$ is a convenience constant. In this case $K_g K_r^2 \to \frac{16}{\pi^3}$, so $K_s \to \infty$, and $D_0 \to \frac{2}{\sqrt{\pi}} \frac{1}{\sqrt{K_g}} \to 3 \pi \gamma \approx 3.3531$. It follows that $L_0' \to \frac{3\pi^2}{2} \gamma \approx 5.2670$.

Let us next be more careful, and not throw away the term $\pi D_0 \frac{\eta}{s}$ and the error terms. $L_0$ as in \eqref{eq:L} is a priori a function of $K_g, K_r$ and $s$. We can take for $K_s$ its exact optimal value according to \eqref{eq:I}, viz.
\begin{align}
K_s = \frac{2 \pi^2 K_r}{\pi^3 K_g K_r^2 - 16}, \label{eq:Ks}
\end{align}
so that $D_0 = \frac{\pi}{2} K_r$. Note that \eqref{eq:Ks} allows us to eliminate from $L_0$ the variable $K_g$. This yields
\begin{align*}
L_0 = \frac{\pi^2}{6} \left(1 + 3\frac{\eta}{s}\right) K_r + \pi^2 \frac{1}{K_r K_s} + 8 \frac{1}{K_r^2}.
\end{align*}

We now minimize $L_0$ by setting the partial derivatives w.r.t.\ $s$ and $K_r$ to $0$. Indeed, $\frac{\partial L_0}{\partial K_r} = \frac{\pi^2}{6}(1 + 3\frac{\eta}{s}) - \pi^2 \frac{1}{K_r^2 K_s} - 16 \frac{1}{K_r^3}$, and this being $0$ implies
\begin{align}
\frac{\pi^2}{6}\left(1 + 3\frac{\eta}{s}\right) K_r^2 - 16 \frac{1}{K_r} = \pi^2 \frac{1}{K_s}. \label{eq:B1}
\end{align}
Further, by $\frac{1}{K_s^2} \frac{dK_s}{ds} = -\frac{(s - 1) e^s + 1}{s^2} c^{-1/3}$ we find $\frac{\partial L_0}{\partial s} = -\frac{\pi^2}{2} \frac{\eta}{s^2} K_r + \pi^2 \frac{1}{K_r} \frac{(s - 1) e^s + 1}{s^2} c^{-1/3}$, and this being $0$ implies
\begin{align}
K_r^2 = \frac{2}{\eta}((s - 1) e^s + 1) c^{-1/3}. \label{eq:B2}
\end{align}

From \eqref{eq:B1} and \eqref{eq:B2} we eliminate $K_r$, and thus obtain an equation in $s$ only, viz.
\begin{align*}
\left(1 + 3\frac{\eta}{s}\right) \frac{1}{\eta^{3/2}} ((s - 1) e^s + 1)^{3/2} - \frac{24\sqrt{2}}{\pi^2} c^{1/2} = 3 \frac{1}{\eta^{1/2}} \frac{e^s - 1 - s}{s} ((s - 1) e^s + 1)^{1/2}.
\end{align*}
The first term on the left hand side is $(\frac{s e^s}{\eta})^{3/2} (1 + O(\frac{1}{s}))$, and the right hand side is $\frac{3(e^s)^{3/2}}{(s \eta)^{1/2}} (1 + O(\frac{1}{s}))$, and as $\eta < 1$ and $s \to \infty$ the right hand side clearly is smaller, so vanishes in the $O(\frac{1}{s})$. So we find $(\frac{s e^s}{\eta})^{3/2} (1 + O(\frac{1}{s})) = \frac{24\sqrt{2}}{\pi^2} c^{1/2}$, and this yields
\begin{align*}
s e^s = \left(\frac{8}{\pi^2\gamma} \eta + O\left(\frac{1}{\log c}\right)\right) c^{1/3}.
\end{align*}
In turn this implies
\begin{align}
s = \frac{1}{3} \log c - \log\log c + \log \eta + O(1), \quad \frac{1}{s} = \frac{3}{\log c} \left(1 + O\left(\frac{\log\log c}{\log c}\right) \right), \label{eq:Ks1}
\end{align}
and
\begin{align}
K_s = \frac{\pi^2 \gamma}{72 \eta} \log^2 c \left(1 + O\left(\frac{1}{\log c}\right)\right). \label{eq:Ks2}
\end{align}
Indeed we find that $K_s$ and $s$ both tend to $\infty$.

To get the proper value for $K_r$ we turn to \eqref{eq:B1}, and introduce $\theta$ such that $K_r = 6 \gamma + \theta$, so that $\theta$ will tend to $0$. Then \eqref{eq:B1} becomes a cubic equation in $\theta$:
\begin{align}
\theta^3 + 18 \gamma \theta^2 + \left(108 \gamma^2 - \frac{6}{\left(1 + 3\frac{\eta}{s}\right) K_s} \right) \theta + \left(\frac{\frac{288}{\pi^2}\frac{\eta}{s}}{1 + 3\frac{\eta}{s}} - \frac{36\gamma}{\left(1 + 3\frac{\eta}{s}\right)K_s}\right) = 0. \label{eq:Kt}
\end{align}
When $s \to \infty$ and $K_s \to \infty$, this ultimately becomes $\theta (\theta^2 + 18 \gamma \theta + 108 \gamma^2) = 0$, with the quadratic term being positive definite, showing that \eqref{eq:Kt} for finite large $s$ has exactly one real solution, which will be close to $0$.
For this solution we have, using \eqref{eq:Ks1}, \eqref{eq:Ks2},
\begin{align*}
\left(108 \gamma^2 + O\left(\frac{1}{\log^2 c}\right) \right) \theta + O(\theta^2) = -\frac{288}{\pi^2} \frac{\eta}{s} \left(1 + O\left(\frac{1}{\log c}\right)\right),
\end{align*}
hence
\begin{align*}
K_r = 6 \gamma \left(1 - \frac{\eta}{s} \left(1 + O\left(\frac{1}{\log c}\right)\right)\right), \quad K_g = \frac{\gamma}{\pi} \left(1 + 2 \frac{\eta}{s} \left(1 + O\left(\frac{1}{\log c}\right)\right)\right).
\end{align*}
Putting everything together, using \eqref{eq:Ks1}, we find
\begin{align*}
L_0 = \frac{3}{2} \pi^2 \gamma \left(1 + 6 \eta \frac{1}{\log c} (1 + o(1)) \right).
\end{align*}
The result now easily follows.
\qed


\begin{thebibliography}{1}
\bibitem{amiri09} E. Amiri and G. Tardos, ``High Rate Fingerprinting Codes and the Fingerprinting Capacity,", \emph{Proc. }20\emph{th ACM-SIAM Symp. on Discrete Algorithms}, pp. 336--345, 2009.
\bibitem{blayer08} O. Blayer and T. Tassa, ``Improved Versions of Tardos' Fingerprinting Scheme," \emph{Des. Codes Cryptogr.}, vol. 48, no. 1, pp. 79--103, 2008.
\bibitem{huang09a} Y. Huang and P. Moulin, ``Saddle-Point Solution of the Fingerprinting Capacity Game Under the Marking Assumption," \emph{Proc. }2009{ IEEE Int. Symp. Information Theory}, vol. 4, pp. 2256--2260, 2009.
\bibitem{msc11} T. Laarhoven, ``Collusion-Resistant Traitor Tracing Schemes," M.Sc. thesis, Dept. Math. Comp. Sc., Eindhoven Univ. Techn., Eindhoven, The Netherlands, 2011.
\bibitem{nuida09} K. Nuida \emph{et al.}, ``An Improvement of Discrete Tardos Fingerprinting Codes," \emph{Des. Codes Cryptogr.}, vol. 52, no. 3, pp. 339--362, 2009.
\bibitem{skoric08} B. \v{S}kori\'{c} \emph{et al.}, ``Symmetric Tardos Fingerprinting Codes for Arbitrary Alphabet Sizes," \emph{Des. Codes Cryptogr.}, vol. 46, no. 2, pp. 137--166, 2008.
\bibitem{skoric06} B. \v{S}kori\'{c} \emph{et al.}, ``Tardos Fingerprinting is Better Than We Thought," \emph{IEEE Trans. Inform. Theory}, vol. 54, no. 8, pp. 3663--3676, 2008.
\bibitem{tardos03} G. Tardos, ``Optimal Probabilistic Fingerprint Codes," \emph{Proc. }35\emph{th ACM Symp. on Theory of Computing}, 2003, pp. 116--125.
\end{thebibliography}
\end{document}